\documentclass[twoside]{article}%
\usepackage{amssymb}
\usepackage{amsfonts}
\usepackage{amsmath}
\usepackage{graphicx}%
\providecommand{\U}[1]{\protect\rule{.1in}{.1in}}
\newtheorem{theorem}{Theorem}

\newtheorem{corollary}[theorem]{Corollary}

\newtheorem{lemma}[theorem]{Lemma}

\newtheorem{remark}[theorem]{Remark}

\newenvironment{proof}[1][Proof]{\noindent\textbf{#1.} }{\ \rule{0.5em}{0.5em}}
\ifx\pdfoutput\relax\let\pdfoutput=\undefined\fi
\newcount\msipdfoutput
\ifx\pdfoutput\undefined\else
\ifcase\pdfoutput\else
\ifx\paperwidth\undefined\else
\ifdim\paperheight=0pt\relax\else\pdfpageheight\paperheight\fi
\ifdim\paperwidth=0pt\relax\else\pdfpagewidth\paperwidth\fi
\fi\fi\fi
\begin{document}

\title{\textbf{Analytical} \textbf{Blowup Solutions to the Isothermal Euler-Poisson
Equations of Gaseous Stars in }$R^{N}$}
\author{Y\textsc{uen} M\textsc{anwai\thanks{E-mail address: nevetsyuen@hotmail.com }}\\\textit{Department of Applied Mathematics, }\\\textit{The Hong Kong Polytechnic University,}\\\textit{Hung Hom, Kowloon, Hong Kong}}
\date{Revised 01-June-2009}
\maketitle

\begin{abstract}
This article is the continued version of the analytical blowup solutions for
2-dimensional Euler-Poisson equations in \cite{Y1} and \cite{Y2}. With the
extension of the blowup solutions with radial symmetry for the isothermal
Euler-Poisson equations in $R^{2}$, other special blowup solutions in $R^{N}$
with non-radial symmetry are constructed by the separation method.


\end{abstract}

\section{\bigskip Introduction}

The evolution of a self-gravitating fluid (gaseous stars) can be formulated by
the isentropic Euler-Poisson equations of the following form:
\begin{equation}
\left\{
\begin{array}
[c]{rl}%
{\normalsize \rho}_{t}{\normalsize +\nabla\cdot(\rho\vec{u})} &
{\normalsize =}{\normalsize 0,}\\
{\normalsize (\rho\vec{u})}_{t}{\normalsize +\nabla\cdot(\rho\vec{u}%
\otimes\vec{u})+\nabla P} & {\normalsize =}{\normalsize -\rho\nabla\Phi,}\\
{\normalsize \Delta\Phi(t,\vec{x})} & {\normalsize =\alpha(N)}%
{\normalsize \rho,}%
\end{array}
\right.  \label{Euler-Poisson}%
\end{equation}
where $\alpha(N)$ is a constant related to the unit ball in $R^{N}$:
$\alpha(1)=2$; $\alpha(2)=2\pi$ and For $N\geq3,$%
\begin{equation}
\alpha(N)=N(N-2)V(N)=N(N-2)\frac{\pi^{N/2}}{\Gamma(N/2+1)},
\end{equation}
where $V(N)$ is the volume of the unit ball in $R^{N}$ and $\Gamma$ is a Gamma
function. And as usual, $\rho=\rho(t,\vec{x})$ and $\vec{u}=\vec{u}(t,\vec
{x})=(u_{1},u_{2},....,u_{N})\in\mathbf{R}^{N}$ are the density, the velocity
respectively. $P=P(\rho)$\ is the pressure.

In the above system, the self-gravitational potential field $\Phi=\Phi
(t,\vec{x})$\ is determined by the density $\rho$ through the Poisson equation.

The equation (\ref{Euler-Poisson})$_{3}$ is the Poisson equation through which
the gravitational potential is determined by the density distribution of the
density itself. Thus, we call the system (\ref{Euler-Poisson}) the
Euler-Poisson equations. The equations can be viewed as a prefect gas model.
The function $P=P(\rho)$\ is the pressure. The $\gamma$-law can be applied on
the pressure $P(\rho)$, i.e.%
\begin{equation}
{\normalsize P}\left(  \rho\right)  {\normalsize =K\rho}^{\gamma}%
:=\frac{{\normalsize \rho}^{\gamma}}{\gamma}, \label{gamma}%
\end{equation}
which is a commonly the hypothesis. The constant $\gamma=c_{P}/c_{v}\geq1$,
where $c_{P}$, $c_{v}$\ are the specific heats per unit mass under constant
pressure and constant volume respectively, is the ratio of the specific heats,
that is, the adiabatic exponent in (\ref{gamma}). In particular, the fluid is
called isothermal if $\gamma=1$. It can be used for constructing models with
non-degenerate isothermal cores, which have a role in connection with the
so-called Schonberg-Chandrasekhar limit \cite{KW}. And we denote the radial
diameter as: $r:=\sqrt{\sum_{j=1}^{N}x_{j}^{2}}$.

The system can be rewritten as%
\begin{equation}
\left\{
\begin{array}
[c]{rl}%
\rho_{t}+\nabla\cdot\vec{u}\rho+\nabla\rho\cdot\vec{u} & {\normalsize =}%
{\normalsize 0,}\\
\rho\left(  \frac{\partial u_{i}}{\partial t}+\sum_{k=1}^{N}u_{k}%
\frac{\partial u_{i}}{\partial x_{k}}\right)  +\frac{\partial}{\partial x_{i}%
}P(\rho) & {\normalsize =}{\normalsize -\rho}\frac{\partial}{\partial x_{i}%
}{\normalsize \Phi(\rho),}\text{ for }i=1,2,...N,\\
{\normalsize \Delta\Phi(t,x)} & {\normalsize =\alpha(N)}{\normalsize \rho.}%
\end{array}
\right.  \label{eqeq1}%
\end{equation}
For $N=3$, (\ref{eqeq1}) is a classical (non-relativistic) description of a
galaxy, in astrophysics. See \cite{BT}, \cite{C} and \cite{KW} for a detail
about the system.

For the local existence results about the system were shown in \cite{B} and
\cite{G}.\ Historically in astrophysics, Goldreich and Weber constructed the
analytical blowup (collapsing) solutions of the $3$-dimensional Euler-Poisson
equations for $\gamma=4/3$ for the non-rotating gas spheres \cite{GW}. After
that, Makino \cite{M1} obtained the rigorously mathematical proof of the
existence of such kind of blowup solutions. Besides, Deng, Xiang and Yang
extended the above blowup solutions in $R^{N}$ ($N\geq4$) \cite{DXY}.
Recently, Yuen obtained the blowup solutions in $R^{2}$ with $\gamma=1$ by a
new transformation \cite{Y1}. The family of the analytical solutions with
radial symmetry
\[
\vec{u}=\frac{\vec{x}}{\left\vert r\right\vert }V(t,r),
\]
are rewritten as

For $N\geq3$ and $\gamma=(2N-2)/N$, in \cite{Y}
\begin{equation}
\left\{
\begin{array}
[c]{c}%
\rho(t,r)=\left\{
\begin{array}
[c]{c}%
\dfrac{1}{a^{N}(t)}y(\frac{r}{a(t)})^{N/(N-2)},\text{ for }r<a(t)Z_{\mu};\\
0,\text{ for }a(t)Z_{\mu}\leq r.
\end{array}
\right.  \text{, }V{\normalsize (t,r)=}\dfrac{\dot{a}(t)}{a(t)}%
{\normalsize r,}\\
\ddot{a}(t){\normalsize =}\dfrac{-\lambda}{a^{N-1}(t)},\text{ }%
{\normalsize a(0)=a}_{1}\neq0{\normalsize ,}\text{ }\dot{a}(0){\normalsize =a}%
_{2},\\
\ddot{y}(z){\normalsize +}\dfrac{N-1}{z}\dot{y}(z){\normalsize +}\dfrac
{\alpha(N)}{(2N-2)K}{\normalsize y(z)}^{N/(N-2)}{\normalsize =\mu,}\text{
}y(0)=\alpha>0,\text{ }\dot{y}(0)=0,
\end{array}
\right.  \label{solution2}%
\end{equation}
where $\mu=[N(N-2)\lambda]/(2N-2)K$ and the finite $Z_{\mu}$ is the first zero
of $y(z)$;

For $N=2$ and $\gamma=1$, in \cite{Y1}%
\begin{equation}
\left\{
\begin{array}
[c]{c}%
\rho(t,r)=\dfrac{1}{a^{2}(t)}e^{y\left(  r/a(t)\right)  }\text{,
}V{\normalsize (t,r)=}\dfrac{\dot{a}(t)}{a(t)}{\normalsize r;}\\
\ddot{a}(t){\normalsize =}\dfrac{-\lambda}{a(t)},\text{ }{\normalsize a(0)=a}%
_{1}>0{\normalsize ,}\text{ }\dot{a}(0){\normalsize =a}_{2};\\
\ddot{y}(z){\normalsize +}\dfrac{1}{z}\dot{y}(z){\normalsize +\dfrac
{\alpha(2)}{K}e}^{y(z)}{\normalsize =\mu,}\text{ }y(0)=\alpha,\text{ }\dot
{y}(0)=0,
\end{array}
\right.  \label{solution 3}%
\end{equation}
where $K>0$, $\mu=2\lambda/K$ with a sufficiently small $\lambda$ and $\alpha$
are constants.\newline However, all the known solutions are in radial
symmetry. In this paper, we are able to obtain the similar results to the
non-radial symmetric cases for the $2$-dimensional Euler-Poisson equations
(\ref{eqeq1}) in the following theorem.

\begin{theorem}
\label{thm2 copy(2)}For the isothermal Euler-Poisson equations (\ref{eqeq1})
in $R^{N}$, there exists a family of solutions,%
\begin{equation}
\left\{
\begin{array}
[c]{c}%
\rho(t,\vec{x})=\frac{1}{a(t)}e^{-\frac{\Phi\left(  \frac{x_{1}^{2}+x_{2}^{2}%
}{a(t)}\right)  }{K}+C}\text{, }{\normalsize \vec{u}(t,\vec{x})=\dfrac
{\overset{\cdot}{a}(t)}{2a(t)}(}x_{1},x_{2},x_{1}+x_{2},...x_{1}%
+x_{2}){\normalsize ,}\\
a(t)=\frac{1}{4a_{1}}\left(  a_{2}t+2a_{1}\right)  ^{2}\\
s\ddot{\Phi}(s)+\frac{\dot{\Phi}(s)}{2}-\epsilon^{\ast}e^{-\tfrac{\Phi(s)}{K}%
}=0,\text{ }\Phi(0)=\alpha\text{, }\dot{\Phi}(0)=\epsilon^{\ast}%
e^{-\tfrac{\alpha}{K}},
\end{array}
\right.  \label{ss2}%
\end{equation}
where $A,$ $B,$ ${\normalsize a}_{1}\neq0,$ ${\normalsize a}_{2\text{ }},$
$\frac{\alpha(N)e^{C}}{4}=\epsilon^{\ast}>0$, $\alpha$ and $\beta$ are
constants.\newline In particular, $a_{1}>0$ and $a_{2}<0$, the solutions
(\ref{ss2}) blow up in the finite time $T=-a_{2}/a_{1}$.
\end{theorem}

\section{Special Blowup Solutions I}

Before presenting the proof of Theorem \ref{thm2 copy(2)}, we prepare the
following two lemmas.

\begin{lemma}
\label{lem:generalsolutionformasseq copy(1)}For the continuity equation
(\ref{eqeq1})$_{1}$ in $R^{N}$, there exist solutions,%
\begin{equation}
\rho(t,\vec{x})=\frac{f\left(  \dfrac{x_{1}^{2}+x_{2}^{2}}{a(t)}\right)
}{a(t)},\text{ }{\normalsize \vec{u}(t,\vec{x})=\frac{\overset{\cdot}{a}%
(t)}{2a(t)}(}x_{1},x_{2},x_{1}+x_{2},...,x_{1}+x_{2}),
\end{equation}
where the scalar function $f(s)\geq0\in C^{1}$ and $a(t)\neq0\in C^{1}.$
\end{lemma}

\begin{proof}
We plug the solutions (\ref{ss2}) into the continuity equation (\ref{eqeq1}%
)$_{1}$,
\begin{align}
&  \rho_{t}+\nabla\cdot\vec{u}\rho+\nabla\rho\cdot\vec{u}\\
&  =\frac{\partial}{\partial t}\left[  \frac{f\left(  \dfrac{x_{1}^{2}%
+x_{2}^{2}}{a(t)}\right)  }{a(t)}\right]  +\nabla\cdot\frac{\dot{a}(t)}%
{2a(t)}(x_{1},x_{2},x_{1}+x_{2},...,x_{1}+x_{2})\frac{f\left(  \dfrac
{x_{1}^{2}+x_{2}^{2}}{a(t)}\right)  }{a(t)}\\
&  +\nabla\frac{f\left(  \dfrac{x_{1}^{2}+x_{2}^{2}}{a(t)}\right)  }%
{a(t)}\cdot\frac{\dot{a}(t)}{a(t)}(x_{1},x_{1},x_{3}+x_{2},....x_{2}+x_{2})\\
&  =\frac{-\dot{a}(t)}{a^{2}(t)}f\left(  \frac{x_{1}^{2}+x_{2}^{2}}%
{a(t)}\right)  +\frac{1}{a(t)}\frac{\partial}{\partial t}f\left(  \frac
{x_{1}^{2}+x_{2}^{2}}{a(t)}\right)  +\frac{\dot{a}(t)}{2a(t)}\left(
\frac{\partial}{\partial x_{1}}x_{1}+\frac{\partial}{\partial x_{2}}%
x_{2}+\underset{i=3}{\overset{N}{\sum}}\frac{\partial}{\partial x_{i}}%
(x_{1}+x_{2})\right)  \frac{f\left(  \dfrac{x_{2}^{2}+x_{2}^{2}}{a(t)}\right)
}{a(t)}\\
&  +\frac{\dot{a}(t)}{2a(t)}\left[  \frac{\partial}{\partial x_{1}}%
\frac{f\left(  \dfrac{x_{1}^{2}+x_{2}^{2}}{a(t)}\right)  }{a(t)}\cdot
x_{1}+\frac{\partial}{\partial x_{2}}\frac{f\left(  \dfrac{x_{1}^{2}+x_{2}%
^{2}}{a(t)}\right)  }{a(t)}\cdot x_{2}+\underset{i=3}{\overset{N}{\sum}}%
\frac{\partial}{\partial x_{i}}\frac{f\left(  \dfrac{x_{1}^{2}+x_{2}^{2}%
}{a(t)}\right)  }{a(t)}\cdot(x_{1}+x_{2})\right]  \\
&  =\frac{-\dot{a}(t)}{a^{2}(t)}f\left(  \frac{x_{1}^{2}+x_{2}^{2}}%
{a(t)}\right)  -\frac{1}{a(t)}\dot{f}\left(  \frac{x_{1}^{2}+x_{2}^{2}}%
{a(t)}\right)  \frac{(x_{1}^{2}+x_{2}^{2})\dot{a}(t)}{a^{2}(t)}+\frac{\dot
{a}(t)}{a(t)}\frac{f\left(  \dfrac{x_{1}^{2}+x_{2}^{2}}{a(t)}\right)  }%
{a^{2}(t)}\\
&  +\frac{\dot{a}(t)}{2a(t)}\left[  \frac{\dot{f}\left(  \dfrac{x_{1}%
^{2}+x_{2}^{2}}{a(t)}\right)  }{a^{2}(t)}\frac{2x_{1}^{2}}{a(t)}+\frac{\dot
{f}\left(  \dfrac{x_{1}^{2}+x_{2}^{2}}{a(t)}\right)  }{a^{2}(t)}\frac
{2x_{2}^{2}}{a(t)}\right]  \\
&  =0.
\end{align}
The proof is completed.
\end{proof}

The following lemma handles the Poisson equation (\ref{eqeq1})$_{3}$ for our
solutions (\ref{ss2}):

\begin{lemma}
\label{lemma2 copy(1)}The solutions,
\begin{equation}
\rho=\frac{{\normalsize 1}}{a(t)}e^{-\frac{\Phi\left(  \frac{x_{1}^{2}%
+x_{2}^{2}}{a(t)}\right)  }{K}+C},\label{aass1}%
\end{equation}
with the second-order ordinary differential equation:%
\begin{equation}
s\ddot{\Phi}(s)+\frac{\dot{\Phi}(s)}{2}-\epsilon^{\ast}e^{-\tfrac{\Phi(s)}{K}%
}=0,\text{ }\Phi(0)=\alpha\text{, }\dot{\Phi}(0)=\epsilon^{\ast}%
e^{-\tfrac{\alpha}{K}},\label{asasasa1}%
\end{equation}
where $s:=(x_{1}^{2}+x_{2}^{2})/a(t)$ and $C$, $\frac{\alpha(N)e^{C}}%
{4}=\epsilon^{\ast}$ and $\alpha$ are constants,\newline fit into the Poisson
equation (\ref{eqeq1})$_{3}$ in $R^{N}$.
\end{lemma}

\begin{proof}
We check that our potential function $\Phi(t,x_{1,}x_{2})$ satisfies the
Poisson equation (\ref{eqeq1})$_{3}$:%
\begin{align}
&  {\normalsize \Delta\Phi(t,\vec{x})-\alpha(N)}{\normalsize \rho}\\
&  =\nabla\cdot\nabla\Phi\left(  \frac{x_{1}^{2}+x_{2}^{2}}{a(t)}\right)
-\frac{{\normalsize \alpha(N)}}{a(t)}e^{-\frac{^{\Phi\left(  \frac{x_{1}%
^{2}+x_{2}^{2}}{a(t)}\right)  }}{K}+C}\\
&  =\nabla\cdot\left[  \frac{\partial}{\partial x_{1}}\Phi\left(  \frac
{x_{1}^{2}+x_{2}^{2}}{a(t)}\right)  ,\frac{\partial}{\partial x_{2}}%
\Phi\left(  \frac{x_{1}^{2}+x_{2}^{2}}{a(t)}\right)  ,\frac{\partial}{\partial
x_{3}}\Phi\left(  \frac{x_{1}^{2}+x_{2}^{2}}{a(t)}\right)  ,...,\frac
{\partial}{\partial x_{N}}\Phi\left(  \frac{x_{1}^{2}+x_{2}^{2}}{a(t)}\right)
\right]  \\
&  -\frac{{\normalsize \alpha(N)}}{a(t)}e^{-\frac{\Phi\left(  \frac{x_{1}%
^{2}+x_{2}^{2}}{a(t)}\right)  }{K}+C}\\
&  =\nabla\cdot\left[  \dot{\Phi}\left(  \frac{x_{1}^{2}+x_{2}^{2}}%
{a(t)}\right)  \frac{2x_{1}}{a(t)},\dot{\Phi}\left(  \frac{x_{1}^{2}+x_{2}%
^{2}}{a(t)}\right)  \frac{2x_{2}}{a(t)},0\right]  -\frac{{\normalsize \alpha
(N)}}{a(t)}e^{-\frac{\Phi\left(  \frac{x_{1}^{2}+x_{2}^{2}}{a(t)}\right)  }%
{K}+C}\\
&  =\frac{\partial}{\partial x_{1}}\left[  \dot{\Phi}(\frac{x_{1}^{2}%
+x_{2}^{2}}{a(t)})\frac{2x_{1}}{a(t)}\right]  +\frac{\partial}{\partial x_{2}%
}\left[  \dot{\Phi}(\frac{x_{1}^{2}+x_{2}^{2}}{a(t)})\frac{2x_{2}}%
{a(t)}\right]  -\frac{{\normalsize \alpha(N)}}{a^{2}(t)}e^{-\frac{\Phi\left(
\frac{x_{1}^{2}+x_{2}^{2}}{a(t)}\right)  }{K}+C}\\
&  =\ddot{\Phi}(\frac{x_{1}^{2}+x_{2}^{2}}{a(t)})\frac{4x_{1}^{2}}{a^{2}%
(t)}+\dot{\Phi}(\frac{x_{1}^{2}+x_{2}^{2}}{a(t)})\frac{2}{a(t)}\\
&  +\ddot{\Phi}(\frac{x_{1}^{2}+x_{2}^{2}}{a(t)})\frac{4x_{2}^{2}}{a^{2}%
(t)}+\dot{\Phi}(\frac{x_{1}^{2}+x_{2}^{2}}{a(t)})\frac{2}{a(t)}-\frac
{{\normalsize \alpha(N)}}{a(t)}e^{-\frac{\Phi\left(  \frac{x_{1}^{2}+x_{2}%
^{2}}{a(t)}\right)  }{K}+C}\\
&  =\frac{4}{a(t)}\left(  s\ddot{\Phi}(s)+\frac{\dot{\Phi}(s)}{2}-\frac
{\alpha(N)e^{C}}{4}e^{-\tfrac{\Phi(s)}{K}}\right)  ,
\end{align}
where we choose $s:=(x_{1}^{2}+x_{2}^{2})/a(t)$ and the ordinary differential
equation:%
\begin{equation}
s\ddot{\Phi}(s)+\frac{\dot{\Phi}(s)}{2}-\epsilon^{\ast}e^{-\tfrac{\Phi(s)}{K}%
}=0,\text{ }\Phi(0)=\alpha\text{, }\dot{\Phi}(0)=\beta,
\end{equation}
with $\frac{\alpha(N)e^{C}}{4}=\epsilon^{\ast}$, $\alpha$ and $\beta$ are
constants$.$ Therefore, our solutions (\ref{aass1}) satisfy the Poisson
equation (\ref{eqeq1})$_{3}$.

The proof is completed.
\end{proof}

Besides, we need the lemma for stating the property of the function $\Phi(s)$
of the analytical solutions (\ref{asasasa1}). We need the lemma for stating
the property of the function $\Phi(s)$. In particular, the solutions
(\ref{ss2}) in $N$-dimensional case involve the following lemma. The similar
lemma was already given in Lemmas 9 and 10, in \cite{Y1}, by the fixed point
theorem. For the completeness of understanding the whole article, the proof is
also presented here.

\begin{lemma}
\label{lemma2}There exists a sufficiently small $\varepsilon^{\ast}>0$, such
that the ordinary differential equation%
\begin{equation}
\left\{
\begin{array}
[c]{c}%
s\ddot{\Phi}(s){\normalsize +}\frac{\dot{\Phi}(s)}{2}-{\normalsize \varepsilon
}^{\ast}{\normalsize e}^{-\dfrac{\Phi(s)}{K}}{\normalsize =0,}\\
\Phi(0)=\alpha\text{, }\dot{\Phi}(0)=-\varepsilon^{\ast}{\normalsize e}%
^{-\dfrac{\alpha}{K}}%
\end{array}
\right.  \label{SecondorderElliptic}%
\end{equation}
where $K>0$ and $\alpha$ are constants, has a unique solution $\Phi(s)\in
C^{2}[0,\infty)$.
\end{lemma}

\begin{proof}
The lemma can be proved by the fixed point theorem. The equation
(\ref{SecondorderElliptic}) can be rewritten as:%
\begin{align}
s^{1/2}\frac{d}{ds}\left(  s^{1/2}\dot{\Phi}(s)\right)   &
=-{\normalsize \varepsilon}^{\ast}e^{-\tfrac{\Phi(s)}{K}},\\
\frac{d}{ds}\left(  s^{1/2}\dot{\Phi}(s)\right)   &  =\frac
{-{\normalsize \varepsilon}^{\ast}e^{-\tfrac{\Phi(s)}{K}}}{s^{1/2}}.
\end{align}
With the initial conditions: $\Phi(0)=\alpha$ and $\dot{\Phi}%
(0)=-{\normalsize \varepsilon}^{\ast}{\normalsize e}^{-\frac{\alpha}{K}}%
$,\ the equation (\ref{SecondorderElliptic}) is reduced to%
\begin{equation}
\dot{\Phi}(s)=\dfrac{-{\normalsize \varepsilon}^{\ast}}{s^{1/2}}\int_{0}%
^{s}\frac{e^{-\tfrac{\Phi(\tau)}{K}}}{\tau^{1/2}}d\tau.
\end{equation}
Set%
\begin{equation}
{\normalsize f(s,\Phi(s))=}\frac{-{\normalsize \varepsilon}^{\ast}}{s^{1/2}%
}\int_{0}^{s}{\normalsize \frac{e^{-\tfrac{\Phi(\tau)}{K}}}{\tau^{1/2}}d\tau.}%
\end{equation}
For any $s_{0}>0$, we get $f\in C^{1}[0,$ $s_{0}]$. And for any $\Phi_{1,}$
$\Phi_{2}\in C^{2}[0,$ $s_{0}]$, we have,%
\begin{equation}
\left\vert f(s,\Phi_{1}(s))-f(s,\Phi_{2}(s))\right\vert =\frac
{{\normalsize \varepsilon}^{\ast}\left\vert \int_{0}^{s}\left(  \frac
{e^{-\tfrac{\Phi_{2}(\tau)}{K}}-e^{-\tfrac{\Phi_{1}(\tau)}{K}}}{\tau^{1/2}%
}\right)  d\tau\right\vert }{s^{1/2}}.
\end{equation}
As $e^{\Phi}$ is a $C^{1}$ function of $\Phi$, we can show that the function
$e^{\Phi}$, is Lipschitz-continuous. Then we get,%
\begin{equation}
\left\vert f(s,\Phi_{1}(s))-f(s,\Phi_{2}(s)\right\vert =\frac
{{\normalsize \varepsilon}^{\ast}\int_{0}^{s}\left\vert \frac{\left(  \Phi
_{2}(\tau\right)  -\Phi_{1}(\tau)}{\tau^{1/2}}\right\vert d\tau}{Ks^{1/2}}%
\leq\frac{{\normalsize \varepsilon}^{\ast}}{{\normalsize K}}\underset{0\leq
s\leq s_{0}}{\sup}\left\vert \Phi_{1}(\tau)-\Phi_{2}(\tau)\right\vert .
\end{equation}
Let%
\begin{equation}
{\normalsize T\Phi(s)=\alpha+}\int_{0}^{s}{\normalsize f(\tau,\Phi(\tau
))d\tau.}%
\end{equation}
We have $T\Phi\in C[0,$ $s_{0}]$\ and%
\begin{equation}
\left\vert T\Phi_{1}(s)-T\Phi_{2}(s)\right\vert =\left\vert \int_{0}^{s}%
f(\tau,\Phi_{1}(\tau))d\tau-\int_{0}^{s}f(\tau,\Phi_{2}(\tau))d\tau\right\vert
\leq\frac{{\normalsize \varepsilon}^{\ast}}{{\normalsize K}}\underset{0\leq
s\leq s_{0}}{\sup}\left\vert \Phi(s)_{1}-\Phi(s)_{2}\right\vert .
\end{equation}
By choosing the constant $\varepsilon^{\ast}$ suck that $0<\frac
{{\normalsize \varepsilon}^{\ast}}{{\normalsize K}}<1$, this shows that the
mapping $T:C[0,$ $s_{0}]\rightarrow C[0,$ $s_{0}]$, is a contraction with the
sup-norm. By the fixed point theorem, there exists a unique $\Phi(s)\in C[0,$
$s_{0}],$\ such that $T\Phi(s)=\Phi(s)$. \newline It is because that the
chosen constant ${\normalsize \varepsilon}^{\ast}$ is independent of the
variable $x$. Therefore, we have the global unique solution $\Phi(s)\in
C^{2}[0,\infty)$. The proof is completed.
\end{proof}

Now, we are ready to check that the solutions fit into the Euler-Poisson
equations (\ref{eqeq1}).

\begin{proof}
[Proof of Theorem \ref{thm2 copy(2)}]By Lemma
\ref{lem:generalsolutionformasseq copy(1)} and Lemma \ref{lemma2 copy(1)}, the
solutions (\ref{ss2}) satisfy (\ref{eqeq1})$_{1}$ and (\ref{eqeq1})$_{3}$. For
the $x_{i}$-component of the isothermal momentum equations (\ref{eqeq1})$_{3}$
in $R^{N}$ $(N\geq3)$, we have%
\begin{align}
&  \rho\left(  \frac{\partial u_{1}}{\partial t}+\sum_{k=1}^{N}u_{k}%
\frac{\partial u_{1}}{\partial x_{k}}\right)  +\frac{\partial}{\partial
x}K\rho+\rho\frac{\partial\Phi}{\partial x}\\
&  =\rho\left[
\begin{array}
[c]{c}%
\frac{\partial}{\partial t}\frac{\dot{a}(t)}{2a(t)}x_{1}+\frac{\dot{a}%
(t)x_{1}}{2a(t)}\frac{\partial}{\partial x_{1}}\frac{\dot{a}(t)x_{1}}%
{2a(t)}+\frac{\dot{a}(t)x_{2}}{2a(t)}\frac{\partial}{\partial x_{2}}\frac
{\dot{a}(t)x_{1}}{2a(t)}\\
+\underset{i=3}{\overset{N}{\sum}}\frac{\dot{a}(t)(x_{1}+x_{2})}{2a(t)}%
\frac{\partial}{\partial x_{i}}\left(  \frac{\dot{a}(t)x_{1}}{2a(t)}\right)
\end{array}
\right]  \\
&  +K\frac{\partial}{\partial x}\frac{e^{-\frac{\Phi\left(  \frac{x_{1}%
^{2}+x_{2}^{2}}{a(t)}\right)  }{K}+C}}{a(t)}+\rho\dot{\Phi}\left(  \frac
{x_{1}^{2}+x_{2}^{2}}{a(t)}\right)  \frac{A}{a(t)}\\
&  =\rho\left[  \frac{1}{2}\left(  \frac{\ddot{a}(t)}{a(t)}-\frac{\dot{a}%
^{2}(t)}{a^{2}(t)}\right)  x_{1}+\frac{1}{4}\frac{\dot{a}(t)x_{1}}{a(t)}%
\frac{\dot{a}(t)}{a(t)}\right]  \\
&  -K\frac{e^{^{-\frac{\Phi\left(  \frac{x_{1}^{2}+x_{2}^{2}}{a(t)}\right)
}{K}+C}}}{a^{2}(t)}\dot{\Phi}\left(  \frac{x_{1}^{2}+x_{2}^{2}}{a(t)}\right)
\frac{2x_{1}}{Ka(t)}+\rho\dot{\Phi}\left(  \frac{x_{1}^{2}+x_{2}^{2}}%
{a(t)}\right)  \frac{2x_{1}}{a(t)}\\
&  =\rho\left[  \frac{1}{4}\left(  \frac{2\ddot{a}(t)}{a(t)}-\frac{\dot{a}%
^{2}(t)}{a^{2}(t)}\right)  x_{1}\right]  -\rho\dot{\Phi}\left(  \frac
{x_{1}^{2}+x_{2}^{2}}{a(t)}\right)  \frac{2x_{1}}{a(t)}+\rho\dot{\Phi}\left(
\frac{x_{1}^{2}+x_{2}^{2}}{a(t)}\right)  \frac{2x_{1}}{a(t)}\\
&  =0,
\end{align}
where we used%
\begin{equation}
2a(t)\ddot{a}(t)-\dot{a}^{2}(t)=0,\text{ }a(0)=a_{1}>0\text{, }\dot
{a}(0)=a_{2}.
\end{equation}
which is exactly solvable by Maple,%
\begin{equation}
a(t)=\frac{1}{4a_{1}}\left(  a_{2}t+2a_{1}\right)  ^{2}.
\end{equation}
For the $x_{2}$-component of the isothermal momentum equations (\ref{eqeq1}%
)$_{3}$ in $R^{N}$, we have%
\begin{align}
&  \rho\left(  \frac{\partial u_{2}}{\partial t}+\sum_{k=1}^{N}u_{k}%
\frac{\partial u_{2}}{\partial x_{k}}\right)  +\frac{\partial}{\partial x_{2}%
}K\rho+\rho\frac{\partial\Phi}{\partial x_{2}}\\
&  =\rho\left[  \frac{\partial}{\partial t}\frac{\dot{a}(t)x_{2}}{2a(t)}%
+\frac{\dot{a}(t)x_{1}}{2a(t)}\frac{\partial}{\partial x_{1}}\frac{\dot{a}%
(t)}{2a(t)}x_{2}+\frac{\dot{a}(t)}{2a(t)}x_{2}\frac{\partial}{\partial x_{2}%
}\frac{\dot{a}(t)x_{2}}{2a(t)}+\sum_{k=3}^{N}u_{k}\frac{\partial}{\partial
x_{k}}\frac{\dot{a}(t)}{2a(t)}(x_{1}+x_{2})\right]  \\
&  +K\frac{\partial}{\partial x}\frac{e^{-\frac{\Phi\left(  \frac{x_{1}%
^{2}+x_{2}^{2}}{a(t)}\right)  }{K}+C}}{a^{2}(t)}+\rho\frac{\partial\Phi\left(
\frac{x^{2}+x_{2}^{2}}{a(t)}\right)  }{\partial x_{2}}\\
&  =\rho\left[  \frac{1}{2}\left(  \frac{\ddot{a}(t)}{a(t)}-\frac{\dot{a}%
^{2}(t)}{a^{2}(t)}\right)  x_{2}+\frac{\dot{a}(t)x_{2}}{4a(t)}\frac{\dot
{a}(t)}{a(t)}\right]  +0+0\\
&  =\rho\left[  \frac{1}{4}\left(  \frac{2\ddot{a}(t)}{a(t)}-\frac{\dot{a}%
^{2}(t)}{a^{2}(t)}\right)  x_{2}\right]  \\
&  =0.
\end{align}
For the $x_{i}$-component $(i\geq3)$ of the isothermal momentum equations
(\ref{eqeq1})$_{3}$ in $R^{N}$, we have%
\begin{align}
&  \rho\left(  \frac{\partial u_{i}}{\partial t}+\sum_{k=1}^{N}u_{k}%
\frac{\partial u_{i}}{\partial x_{k}}\right)  +\frac{\partial}{\partial x_{i}%
}K\rho+\rho\frac{\partial\Phi}{\partial x_{i}}\\
&  =\rho\left[
\begin{array}
[c]{c}%
\frac{\partial}{\partial t}\frac{\dot{a}(t)(x_{1}+x_{2})}{2a(t)}+\frac{\dot
{a}(t)x_{1}}{2a(t)}\frac{\partial}{\partial x_{1}}\frac{\dot{a}(t)x_{1}%
}{2a(t)}+\frac{\dot{a}(t)x_{2}}{2a(t)}\frac{\partial}{\partial x_{2}}%
\frac{\dot{a}(t)x_{2}}{2a(t)}\\
+\sum_{k=3}^{N}u_{k}\frac{\partial}{\partial x_{k}}\frac{\dot{a}%
(t)(x_{1}+x_{2})}{2a(t)}%
\end{array}
\right]  \\
&  +K\frac{\partial}{\partial x_{i}}\frac{e^{-\frac{\Phi\left(  \frac
{x_{1}^{2}+x_{2}^{2}}{a(t)}\right)  }{K}+C}}{a^{2}(t)}+\rho\frac{\partial
}{\partial x_{i}}\dot{\Phi}\left(  \frac{x_{1}^{2}+x_{2}^{2}}{a(t)}\right)  \\
&  =\rho\left[  \frac{1}{2}\left(  \frac{\ddot{a}(t)}{a(t)}-\frac{\dot{a}%
^{2}(t)}{a^{2}(t)}\right)  (x_{1}+x_{2})+\frac{\dot{a}(t)x_{1}}{2a(t)}%
\frac{\dot{a}(t)}{2a(t)}+\frac{\dot{a}(t)x_{2}}{2a(t)}\frac{\dot{a}(t)}%
{2a(t)}\right]  \\
&  =\rho\left[  \frac{1}{4}\left(  \frac{2\ddot{a}(t)}{a(t)}-\frac{\dot{a}%
^{2}(t)}{a^{2}(t)}\right)  (x_{1}+x_{2})\right]  \\
&  =0.
\end{align}
Therefore, our solutions satisfy the Euler-Poisson equations. In particular,
$a_{1}>0$ and $a_{2}<0$, the solutions (\ref{ss2}) blow up in the finite time
$T=-a_{2}/a_{1}$.

The proof is completed.
\end{proof}

It is clear to see the blowup rate of the solutions (\ref{ss2}):

\begin{corollary}
\label{thm:2 copy(1)}The blowup rate of the solutions (\ref{ss2}) is,%
\begin{equation}
\underset{t\rightarrow T}{\lim}\rho(t,\vec{0})\left(  T-t\right)  ^{2}\geq
O(1).
\end{equation}

\end{corollary}

\section{Special Blowup Solutions II}

In the recent paper \cite{Y2}, we have the special solutions for the
isothermal Euler-Poisson equations in $R^{2},$ in the following from:%
\begin{equation}
\left\{
\begin{array}
[c]{c}%
\rho(t,x_{1},x_{2})=\frac{1}{a^{2}(t)}e^{-\frac{\Phi\left(  \frac
{Ax_{1}+Bx_{2}}{a(t)}\right)  }{K}+C}\text{, }{\normalsize \vec{u}(t,x}%
_{1}{\normalsize ,x}_{2}{\normalsize )=\dfrac{\overset{\cdot}{a}(t)}{a(t)}%
(}x_{1},x_{2}){\normalsize ,}\\
a(t)=a_{1}+a_{2}t,\\
\ddot{\Phi}(s)-\epsilon^{\ast}e^{-\frac{\Phi(s)}{K}}=0,\text{ }\Phi
(0)=\alpha,\text{ }\dot{\Phi}(0)=\beta,
\end{array}
\right.
\end{equation}
where $A,$ $B,$ ${\normalsize a}_{1}\neq0,$ ${\normalsize a}_{2\text{ }},$
$\frac{2\pi e^{C}}{A^{2}+B^{2}}=\epsilon^{\ast}>0$, $\alpha$ and $\beta$ are
constants.\newline

In this section, we extend the above solutions to the $N$-dimensional
Euler-Poisson equations (\ref{eqeq1}) in the following theorem:

\begin{theorem}
\label{thm2}For the isothermal Euler-Poisson equations (\ref{eqeq1}) in
$R^{N}$ $(N\geq3)$, there exists a family of solutions,%
\begin{equation}
\left\{
\begin{array}
[c]{c}%
\rho(t,\vec{x})=\frac{1}{a^{2}(t)}e^{-\frac{\Phi\left(  \frac{Ax_{1}+Bx_{2}%
}{a(t)}\right)  }{K}+C}\text{, }{\normalsize \vec{u}(t,\vec{x})=\dfrac
{\overset{\cdot}{a}(t)}{a(t)}(}x_{1},x_{2},x_{1},...,x_{1}){\normalsize ,}\\
a(t)=a_{1}+a_{2}t,\\
\ddot{\Phi}(s)-\epsilon^{\ast}e^{-\frac{\Phi(s)}{K}}=0,\text{ }\Phi
(0)=\alpha,\text{ }\dot{\Phi}(0)=\beta,
\end{array}
\right.  \label{ss21}%
\end{equation}
where $A,$ $B,$ ${\normalsize a}_{1}\neq0,$ ${\normalsize a}_{2\text{ }},$
$\frac{\alpha(N)e^{C}}{A^{2}+B^{2}}=\epsilon^{\ast}>0$, $\alpha$ and $\beta$
are constants.\newline In particular, $a_{1}>0$ and $a_{2}<0$, the solutions
(\ref{ss21}) blow up in the finite time $T=-a_{2}/a_{1}$.
\end{theorem}

Before presenting the proof of Theorem \ref{thm2 copy(2)}, we prepare the
following two lemmas.

\begin{lemma}
\label{lem:generalsolutionformasseq copy(1)}For the continuity equation
(\ref{eqeq1})$_{1}$ in $R^{N}$, there exist solutions,%
\begin{equation}
\rho(t,\vec{x})=\frac{f\left(  \dfrac{Ax_{1}+Bx_{2}}{a(t)}\right)  }{a^{2}%
(t)},\text{ }{\normalsize \vec{u}(t,\vec{x})=\frac{\overset{\cdot}{a}%
(t)}{a(t)}(}x_{1},x_{2},x_{1},...,x_{1}),
\end{equation}
where the scalar function $f(s)\geq0\in C^{1}$ and $a(t)\neq0\in C^{1}.$
\end{lemma}

\begin{proof}
We plug the solutions (\ref{ss21}) into the continuity equation (\ref{eqeq1}%
)$_{1}$,
\begin{align}
&  \rho_{t}+\nabla\cdot\vec{u}\rho+\nabla\rho\cdot\vec{u}\\
&  =\frac{\partial}{\partial t}\left[  \frac{f\left(  \dfrac{Ax_{1}+Bx_{2}%
}{a(t)}\right)  }{a^{2}(t)}\right]  +\nabla\cdot\frac{\dot{a}(t)}{a(t)}%
(x_{1},x_{2},x_{1},....,x_{1})\frac{f\left(  \dfrac{Ax_{1}+Bx_{2}}%
{a(t)}\right)  }{a^{2}(t)}\\
&  +\nabla\frac{f\left(  \dfrac{Ax_{1}+Bx_{2}}{a(t)}\right)  }{a^{2}(t)}%
\cdot\frac{\dot{a}(t)}{a(t)}(x_{1},x_{2},x_{1},...,x_{1})\\
&  =\frac{-2\dot{a}(t)}{a^{3}(t)}f\left(  \frac{Ax_{1}+Bx_{2}}{a(t)}\right)
+\frac{1}{a^{2}(t)}\frac{\partial}{\partial t}f\left(  \frac{Ax_{1}+Bx_{2}%
}{a(t)}\right)  +\frac{\dot{a}(t)}{a(t)}\left(  \frac{\partial}{\partial
x_{1}}x_{1}+\frac{\partial}{\partial x_{2}}x_{2}+\overset{N}{\underset
{i=3}{\sum}}\frac{\partial}{\partial x_{i}}x_{1}\right)  \frac{f\left(
\dfrac{Ax_{1}+Bx_{2}}{a(t)}\right)  }{a^{2}(t)}\\
&  +\frac{\dot{a}(t)}{a(t)}\left[  \frac{\partial}{\partial x_{1}}%
\frac{f\left(  \dfrac{Ax_{1}+Bx_{2}}{a(t)}\right)  }{a^{2}(t)}\cdot
x_{1}+\frac{\partial}{\partial x_{2}}\frac{f\left(  \dfrac{Ax_{1}+Bx_{2}%
}{a(t)}\right)  }{a^{2}(t)}\cdot x_{2}+\overset{N}{\underset{i=3}{\sum}}%
\frac{\partial}{\partial x_{i}}\frac{f\left(  \dfrac{Ax_{1}+Bx_{2}}%
{a(t)}\right)  }{a^{2}(t)}\cdot x_{1}\right]  \\
&  =\frac{-2\dot{a}(t)}{a^{3}(t)}f\left(  \frac{Ax_{1}+Bx_{2}}{a(t)}\right)
-\frac{1}{a^{2}(t)}\dot{f}\left(  \frac{Ax_{1}+Bx_{2}}{a(t)}\right)
\frac{(Ax_{1}+Bx_{2})\dot{a}(t)}{a^{2}(t)}+2\frac{\dot{a}(t)}{a(t)}%
\frac{f\left(  \dfrac{Ax_{1}+Bx_{2}}{a(t)}\right)  }{a^{2}(t)}\\
&  +\frac{\dot{a}(t)}{a(t)}\left[  \frac{\dot{f}\left(  \dfrac{Ax_{1}+Bx_{2}%
}{a(t)}\right)  }{a^{2}(t)}\frac{Ax_{1}}{a(t)}+\frac{\dot{f}\left(
\dfrac{Ax_{1}+Bx_{2}}{a(t)}\right)  }{a^{2}(t)}\frac{Bx_{2}}{a(t)}\right]  \\
&  =0.
\end{align}
The proof is completed.
\end{proof}

The following lemma handles the Poisson equation (\ref{eqeq1})$_{3}$ for our
solutions (\ref{ss21}):

\begin{lemma}
\label{lemma2 copy(1)}The solutions,
\begin{equation}
\rho=\frac{{\normalsize 1}}{a^{2}(t)}e^{-\frac{\Phi\left(  \frac{Ax_{1}%
+Bx_{2}}{a(t)}\right)  }{K}+C},\label{aass1}%
\end{equation}
with the second-order ordinary differential equation:%
\begin{equation}
\ddot{\Phi}(s)-\epsilon^{\ast}e^{-\tfrac{\Phi(s)}{K}}=0,\text{ }\Phi
(0)=\alpha\text{, }\dot{\Phi}(0)=\beta,
\end{equation}
where $s:=(Ax_{1}+Bx_{2})/a(t)$ and $C$, $\frac{\alpha(N)e^{C}}{A^{2}+B^{2}%
}=\epsilon^{\ast}$, $\alpha$ and $\beta$ are constants,\newline fit into the
Poisson equation (\ref{eqeq1})$_{3}$ in $R^{N}$.
\end{lemma}

\begin{proof}
We check that our potential function $\Phi(t,x_{1},x_{2})$ satisfies the
Poisson equation (\ref{eqeq1})$_{3}$:%
\begin{align}
&  {\normalsize \Delta\Phi(t,x}_{1}{\normalsize ,x}_{2}{\normalsize )-\alpha
(N)}{\normalsize \rho}\\
&  =\nabla\cdot\nabla\Phi\left(  \frac{Ax_{1}+Bx_{2}}{a(t)}\right)
-\frac{{\normalsize \alpha(N)}}{a^{2}(t)}e^{-\frac{^{\Phi\left(  \frac
{Ax_{1}+Bx_{2}}{a(t)}\right)  }}{K}+C}\\
&  =\nabla\cdot\left[  \frac{\partial}{\partial x_{1}}\Phi\left(  \frac
{Ax_{1}+Bx_{2}}{a(t)}\right)  ,\frac{\partial}{\partial x_{2}}\Phi\left(
\frac{Ax_{1}+Bx_{2}}{a(t)}\right)  ,\frac{\partial}{\partial x_{3}}\Phi\left(
\frac{Ax_{1}+Bx_{2}}{a(t)}\right)  ,...,\frac{\partial}{\partial x_{N}}%
\Phi\left(  \frac{Ax_{1}+Bx_{2}}{a(t)}\right)  \right]  \\
&  -\frac{{\normalsize \alpha(N)}}{a^{2}(t)}e^{-\frac{\Phi\left(  \frac
{Ax_{1}+Bx_{2}}{a(t)}\right)  }{K}+C}\\
&  =\nabla\cdot\left[  \dot{\Phi}\left(  \frac{Ax_{1}+Bx_{2}}{a(t)}\right)
\frac{A}{a(t)},\dot{\Phi}\left(  \frac{Ax_{1}+Bx_{2}}{a(t)}\right)  \frac
{B}{a(t)},\vec{0}\right]  -\frac{{\normalsize \alpha(N)}}{a^{2}(t)}%
e^{-\frac{\Phi\left(  \frac{Ax_{1}+Bx_{2}}{a(t)}\right)  }{K}+C}\\
&  =\frac{\partial}{\partial x_{1}}\left[  \dot{\Phi}(\frac{Ax_{1}+Bx_{2}%
}{a(t)})\frac{A}{a(t)}\right]  +\frac{\partial}{\partial x_{2}}\left[
\dot{\Phi}(\frac{Ax_{1}+Bx_{2}}{a(t)})\frac{B}{a(t)}\right]  -\frac
{{\normalsize \alpha(N)}}{a^{2}(t)}e^{-\frac{\Phi\left(  \frac{Ax_{1}+Bx_{2}%
}{a(t)}\right)  }{K}+C}\\
&  =\frac{A^{2}+B^{2}}{a^{2}(t)}\left(  \ddot{\Phi}(s)-\frac{\alpha(N)e^{C}%
}{A^{2}+B^{2}}e^{-\tfrac{\Phi(s)}{K}}\right)  ,
\end{align}
where we choose $s:=(Ax_{1}+Bx_{2})/a(t)$ and the ordinary differential
equation:%
\begin{equation}
\ddot{\Phi}(s)-\epsilon^{\ast}e^{-\tfrac{\Phi(s)}{K}}=0,\text{ }\Phi
(0)=\alpha\text{, }\dot{\Phi}(0)=\beta,
\end{equation}
with $\frac{\alpha(N)e^{C}}{A^{2}+B^{2}}=\epsilon^{\ast}$, $\alpha$ and
$\beta$ are constants$.$ Therefore, our solutions (\ref{aass1}) satisfy the
Poisson equation (\ref{eqeq1})$_{3}$.

The proof is completed.
\end{proof}

Now, we are ready to check that the solutions fit into the Euler-Poisson
equations (\ref{eqeq1}).

\begin{proof}
[Proof of Theorem \ref{thm2 copy(2)}]By Lemma
\ref{lem:generalsolutionformasseq copy(1)} and Lemma \ref{lemma2 copy(1)}, the
solutions (\ref{ss21}) satisfy (\ref{eqeq1})$_{1}$ and (\ref{eqeq1})$_{3}$.
For the $x_{1}$-component of the isothermal momentum equations (\ref{eqeq1}%
)$_{3}$ in $R^{N}$ $(N\geq3)$, we have%
\begin{align}
&  \rho\left(  \frac{\partial u_{1}}{\partial t}+\sum_{k=1}^{N}u_{k}%
\frac{\partial u_{1}}{\partial x_{k}}\right)  +\frac{\partial}{\partial x_{1}%
}K\rho+\rho\frac{\partial\Phi}{\partial x_{1}}\\
&  =\rho\left[  \frac{\partial}{\partial t}\frac{\dot{a}(t)x_{1}}{a(t)}%
+\frac{\dot{a}(t)x_{1}}{a(t)}\frac{\partial}{\partial x_{1}}\frac{\dot
{a}(t)x_{1}}{a(t)}+\frac{\dot{a}(t)x_{2}}{a(t)}\frac{\partial}{\partial x_{2}%
}\frac{\dot{a}(t)x_{1}}{a(t)}+\sum_{k=3}^{N}u_{k}\frac{\partial}{\partial
x_{k}}\frac{\dot{a}(t)x_{1}}{a(t)}\right]  \\
&  +K\frac{\partial}{\partial x_{1}}\frac{e^{-\frac{\Phi\left(  \frac
{Ax_{1}+Bx_{2}}{a(t)}\right)  }{K}+C}}{a^{2}(t)}+\rho\dot{\Phi}\left(
\frac{Ax_{1}+Bx_{2}}{a(t)}\right)  \frac{A}{a(t)}\\
&  =\rho\left[  \left(  \frac{\ddot{a}(t)}{a(t)}-\frac{\dot{a}^{2}(t)}%
{a^{2}(t)}\right)  x_{1}+\frac{\dot{a}(t)x_{1}}{a(t)}\frac{\dot{a}(t)}%
{a(t)}\right]  \\
&  -K\frac{e^{^{-\frac{\Phi\left(  \frac{Ax_{1}+Bx_{2}}{a(t)}\right)  }{K}+C}%
}}{a^{2}(t)}\dot{\Phi}\left(  \frac{Ax_{1}+Bx_{2}}{a(t)}\right)  \frac
{A}{Ka(t)}+\rho\dot{\Phi}\left(  \frac{Ax_{1}+Bx_{2}}{a(t)}\right)  \frac
{A}{a(t)}\\
&  =\rho\left[  \frac{\ddot{a}(t)x_{1}}{a(t)}\right]  -\rho\dot{\Phi}\left(
\frac{Ax_{1}+Bx_{2}}{a(t)}\right)  \frac{A}{a(t)}+\rho\dot{\Phi}\left(
\frac{Ax_{1}+Bx_{2}}{a(t)}\right)  \frac{A}{a(t)}\\
&  =0,
\end{align}
by taking $\ddot{a}(t)=0$ that is%
\begin{equation}
a(t)=a_{1}+a_{2}t.
\end{equation}
For the $x_{2}$-component of the isothermal momentum equations (\ref{eqeq1}%
)$_{3}$ in $R^{N}$, we have%
\begin{align}
&  \rho\left(  \frac{\partial u_{2}}{\partial t}+\sum_{k=1}^{N}u_{k}%
\frac{\partial u_{2}}{\partial x_{k}}\right)  +\frac{\partial}{\partial x_{2}%
}K\rho+\rho\frac{\partial\Phi}{\partial x_{2}}\\
&  =\rho\left[  \frac{\partial}{\partial t}\frac{\dot{a}(t)x_{2}}{a(t)}%
+\frac{\dot{a}(t)x_{1}}{a(t)}\frac{\partial}{\partial x_{1}}\frac{\dot
{a}(t)x_{2}}{a(t)}+\frac{\dot{a}(t)x_{2}}{a(t)}\frac{\partial}{\partial x_{2}%
}\frac{\dot{a}(t)x_{2}}{a(t)}+\sum_{k=3}^{N}u_{k}\frac{\partial}{\partial
x_{k}}\frac{\dot{a}(t)x_{2}}{a(t)}\right]  \\
&  +K\frac{\partial}{\partial x_{2}}\frac{e^{-\frac{\Phi\left(  \frac
{Ax_{1}+Bx_{2}}{a(t)}\right)  }{K}+C}}{a^{2}(t)}+\rho\frac{\partial\Phi\left(
\frac{Ax_{1}+Bx_{2}}{a(t)}\right)  }{\partial x_{2}}\\
&  =\rho\left[  \left(  \frac{\ddot{a}(t)}{a(t)}-\frac{\dot{a}^{2}(t)}%
{a^{2}(t)}\right)  x_{2}+\frac{\dot{a}(t)x_{2}}{a(t)}\frac{\dot{a}(t)}%
{a(t)}\right]  \\
&  -K\frac{e^{^{-\frac{\Phi\left(  \frac{Ax_{1}+Bx_{2}}{a(t)}\right)  }{K}+C}%
}}{a^{2}(t)}\dot{\Phi}\left(  \frac{Ax_{1}+Bx_{2}}{a(t)}\right)  \frac
{B}{Ka(t)}+\rho\dot{\Phi}\left(  \frac{Ax_{1}+Bx_{2}}{a(t)}\right)  \frac
{B}{a(t)}\\
&  =\rho\left[  \frac{\ddot{a}(t)x_{2}}{a(t)}\right]  \\
&  =0.
\end{align}
For the $x_{i}(i\geq3)$-component of the isothermal momentum equations
(\ref{eqeq1})$_{i}$ in $R^{N}$, we have%
\begin{align}
&  \rho\left(  \frac{\partial u_{3}}{\partial t}+\sum_{k=1}^{N}u_{k}%
\frac{\partial u_{3}}{\partial x_{k}}\right)  +\frac{\partial}{\partial x_{i}%
}K\rho+\rho\frac{\partial\Phi}{\partial x_{i}}\\
&  =\rho\left[  \frac{\partial}{\partial t}\frac{\dot{a}(t)x_{1}}{a(t)}%
+\frac{\dot{a}(t)x_{1}}{a(t)}\frac{\partial}{\partial x_{1}}\frac{\dot
{a}(t)x_{1}}{a(t)}+\frac{\dot{a}(t)x_{2}}{a(t)}\frac{\partial}{\partial x_{2}%
}\frac{\dot{a}(t)x_{1}}{a(t)}+\sum_{k=1}^{N}u_{k}\frac{\partial}{\partial
x_{k}}\frac{\dot{a}(t)x_{1}}{a(t)}\right]  +0+0\\
&  =\rho\left[  \left(  \frac{\ddot{a}(t)}{a(t)}-\frac{\dot{a}^{2}(t)}%
{a^{2}(t)}\right)  x_{1}+\frac{\dot{a}(t)x_{1}}{a(t)}\frac{\dot{a}(t)}%
{a(t)}\right]  \\
&  =\rho\frac{\ddot{a}(t)}{a(t)}x_{1}\\
&  =0.
\end{align}
Therefore, our solutions satisfy the Euler-Poisson equations. In particular,
$a_{1}>0$ and $a_{2}<0$, the solutions (\ref{ss21}) blow up in the finite time
$T=-a_{2}/a_{1}$.

The proof is completed.
\end{proof}

\begin{remark}
The existence and uniqueness of the function $\Phi(s)$ in the solutions
(\ref{ss21})$_{3}$ can be shown by the fixed point theorem, if $\epsilon
^{\ast}$ is a sufficient small number.
\end{remark}

Additionally the blowup rate about the solutions is immediately followed:

\begin{corollary}
\label{thm:2}The blowup rate of the solutions (\ref{ss21}) is,%
\begin{equation}
\underset{t\rightarrow T}{\lim}\rho(t,\vec{0})\left(  T-t\right)  ^{2}\geq
O(1).
\end{equation}

\end{corollary}

\begin{remark}
The other analytical solutions in $R^{N}$ can be constructed by
\end{remark}%

\begin{equation}
\left\{
\begin{array}
[c]{c}%
\rho(t,\vec{x})=\frac{1}{a^{2}(t)}e^{-\frac{\Phi\left(  \frac{Ax_{1}+Bx_{2}%
}{a(t)}\right)  }{K}+C}\text{, }{\normalsize \vec{u}(t,\vec{x})=\dfrac
{\overset{\cdot}{a}(t)}{a(t)}(}x_{1},x_{2},\tilde{x}_{3},...,\tilde{x}%
_{N}){\normalsize ,}\\
a(t)=a_{1}+a_{2}t,\\
\ddot{\Phi}(s)-\epsilon^{\ast}e^{-\frac{\Phi(s)}{K}}=0,\text{ }\Phi
(0)=\alpha,\text{ }\dot{\Phi}(0)=\beta,
\end{array}
\right.  \label{kk0}%
\end{equation}

where $\tilde{x}_{i}=x_{1}$ or $x_{2}$.

\begin{remark}
Our solutions (\ref{ss2}), (\ref{ss21}) and (\ref{kk0}) also work for the
isothermal Navier-Stokes-Poisson equations in $R^{N}$ $(N\geq3)$:
\begin{equation}
\left\{
\begin{array}
[c]{rl}%
{\normalsize \rho}_{t}{\normalsize +\nabla\cdot(\rho\vec{u})} &
{\normalsize =}{\normalsize 0,}\\
{\normalsize (\rho\vec{u})}_{t}{\normalsize +\nabla\cdot(\rho\vec{u}%
\otimes\vec{u})+\nabla K\rho} & {\normalsize =}{\normalsize -\rho\nabla
\Phi+\mu\Delta\vec{u},}\\
{\normalsize \Delta\Phi(t,x)} & {\normalsize =\alpha(N)}{\normalsize \rho,}%
\end{array}
\right.
\end{equation}
where $\mu>0$ is a positive constant.
\end{remark}

In conclusion, due to the novel solutions obtained by the separation method,
the author conjectures there exists other analytical solution in non-radial
symmetry. Further works will be continued for seeking more particular
solutions to understand the nature of the Euler-Poisson equations (\ref{eqeq1}).

\end{document}